\documentclass[11pt,a4paper]{article}

\usepackage[utf8]{inputenc} 
\usepackage[T1]{fontenc}    %
\usepackage[english]{babel} 
\usepackage[final]{microtype} 

\usepackage{amssymb,amsmath,mathtools} 
\usepackage{amsthm} 

\usepackage{xcolor}
\usepackage{bbm}
\usepackage{braket}

\usepackage[hmargin=0.12\paperwidth,vmargin=0.16\paperwidth,bindingoffset=0cm]%
{geometry} 

\numberwithin{equation}{section} 

\theoremstyle{plain}
  
  \newtheorem{prop}{Proposition}
  
  \newtheorem{cor}{Corollary}
\theoremstyle{definition}
  
  \newtheorem{remark}{Remark}
  
  \numberwithin{prop}{section}
   \numberwithin{cor}{section}
   \numberwithin{remark}{section}

\title{\Large\bfseries Dualities for characteristic polynomial averages of complex symmetric and self dual non-Hermitian random matrices}%
   
\author{Peter J. Forrester}
\date{}

\begin{document}

\maketitle

School of Mathematics and Statistics,  The University of Melbourne,
Victoria 3010, Australia. \: \: Email: {\tt pjforr@unimelb.edu.au}; \\

\bigskip

\begin{abstract}
\noindent
Ensembles of complex symmetric, and complex self dual random matrices are known to exhibit
local statistical properties distinct from those of the non-Hermitian Ginibre ensembles. On the other
hand, in distinction to the latter, the joint eigenvalue probability density function of these two
ensembles are not known. Nonetheless, as carried out in the recent works of Liu and Zhang,
Akemann et al.~and Kulkarni et al., by considering averages of products of characteristic polynomials, analytic progress
can be made. Here we show that an approach based on the theory of zonal polynomials provides
an alternative to the diffusion operator or supersymmetric Grassmann integrations methods of these works. It has
the advantage of not being restricted to a Gaussian unitary invariant measure on the matrix spaces. To illustrate this, as an extension,
we consider averages of  products and powers of characteristic polynomials for complex symmetric, and complex self dual random matrices
subject to a  spherical measure. In the case of powers, when comparing against the corresponding real Ginibre, respectively quaternion Ginibre
averages with a spherical measure, one finds the qualitative feature of a decreasing (increasing)
profile as the magnitude of the argument of the characteristic polynomial increases. This is analogous to
the findings of the second two of the cited works in the Gaussian case.

\end{abstract}

\vspace{3em}

\section{Introduction}
The three Ginibre ensembles of non-Hermitian $N \times N$ random matrices all consist of independent standard
Gaussians. They are distinguished by the number field corresponding to the entries, this being real (GinOE),
complex (GinUE) or quaternion (GinSE); see the recent text \cite{BF24} in relation to the naming and further
background material. In the quaternion case each element is represented as a particular $2 \times 2$ complex
matrix (see \cite[Eq.(1.1)]{BF24}), so as a complex matrix the size is $2N \times 2N$. An implied symmetry relating to
complex conjugation then requires
that the eigenvalues  occur in complex conjugate pairs. For a different reason -- simply that the characteristic polynomial
has all coefficients real -- this is similarly true for the GinOE, although now some of the eigenvalues will also be real
\cite{Ed97}.

Let $X$ be a Ginibre matrix from one of the three cases and form the Hermitian matrix $H = {1 \over 2} (X + X^\dagger)$. 
The corresponding ensembles are then the GOE (real symmetric matrices), GUE (complex Hermitian matrices)
and GSE (self dual Hermitian quaternion matrices) identified in pioneering work in the field due to Wigner and Dyson;
see the compilation of reprints of this early literature and associated commentary in \cite{Po65}. The joint eigenvalue
probability density function (PDF)  is then proportional to (see e.g.~\cite[Prop.~1.3.4]{Fo10})
  \begin{equation}\label{A.1}
  \prod_{l=1}^N e^{-\beta \lambda_l^2/2} \prod_{1 \le j < k \le N} |\lambda_k - \lambda_j|^\beta.
\end{equation}
Here the parameter $\beta$ --- known as the Dyson index --- takes on the values $\beta = 1$ (GOE),
$\beta = 2$ (GUE) and $\beta = 4$ (GSE). In the bulk scaled large $N$ limit (mean eigenvalue
spacing unity) the Dyson index determines the small distance behaviour of the PDF for the
spacing distance $s$ between consecutive eigenvalues, $P(s) \mathop{\sim}\limits_{s \to 0^+} c_\beta s^{\beta}$
for some (known \cite[Eq.~(8.165)]{Fo10}) constant $c_\beta$. The Dyson index also distinguishes various  other
statistical quantities. Consider for example the linear statistic $ \sum_{j=1}^N f(x_j)$ where $x_j = \lambda_j/\sqrt{2N}$.
The scaling determining $\{x_j\}$ has been chosen so that to leading order the spectrum is contained in the interval
$(-1,1)$; see e.g.~\cite[\S 1.4.2]{Fo10}. Assume $f(x)$ is smooth and introduce the expansion
$f(\cos \phi) = f_0 + \sum_{k=1}^\infty f_k \cos k \phi$, $f_k = {2 \over \pi} \int_0^\pi f(\cos \phi) \cos k \phi \, d \phi$ for
$k \ge 1$. Then one has that (see e.g.~\cite[Eq.~(3.2)]{Fo23})
  \begin{equation}\label{A.2}
 \lim_{N \to \infty} {\rm Var}\Big (   \sum_{j=1}^N f(x_j) \Big ) = {2 \over \beta} \sum_{k=1}^\infty k f_k^2.
 \end{equation}
 
 An immediate point of interest is to ask to what extent the Dyson parameter of the Gaussian Hermitian ensembles
 carries over to their non-Hermitian Ginibre counterparts? It remains true that $\beta$ can be used to distinguish the
 number field of the entries. Specifically in this sense $\beta$ specifies the number of real numbers required to specify
 a general entry, and so one associates $\beta = 1,2,4$ with the GinOE, GinUE and GinSE respectively. More 
 significant is that with this labelling, the analogue of (\ref{A.2}) is known to hold true. One introduces the linear
 statistic  $\sum_{j=1}^N f(\mathbf r_j)$ where $\mathbf r_j = \lambda_j/\sqrt{N}$, with the scaling being chosen so that to leading
 order the eigenvalue support is the unit disk in the complex plane. Define Var$(\sum_{j=1}^N f(\mathbf r_j)) = \langle
 |\sum_{j=1}^N f(\mathbf r_j)|^2 \rangle - |  \langle \sum_{j=1}^N f(\mathbf r_j) \rangle |^2 $ (here the use of the absolute values
 allows for the case that $f$ is complex).
 Using the just specified Dyson index labelling,
 and assuming that $f$ is smooth, the corresponding variance is given by 
  \begin{equation}\label{A.2}
 \lim_{N \to \infty} {\rm Var}\Big (   \sum_{j=1}^N f(\mathbf r_j) \Big ) = {1 \over 2 \pi \beta} \int_{|\mathbf r| < 1}
 | \Delta f|^2 \, d \mathbf r + {1 \over \beta} 
  \sum_{k=-\infty}^\infty |k| f_k f_{-k}, 
 \end{equation}  
 where
  \begin{equation}\label{A.2a}
   f(\mathbf r) \Big |_{\mathbf r = (\cos \theta, \sin \theta)} = \sum_{n=-\infty}^\infty f_n e^{in \theta};
 \end{equation}    
   see \cite{RV06,AHM11} in relation to GinUE, \cite{CES21a} in relation to GinOE, \cite{Ko15,OR16} for an analysis
 applying to both the GinUE and GinOE, and \cite{BF23} in relation to the GinSE.
In the real $(\beta = 1$) and quaternion $(\beta = 4)$ cases, when the eigenvalues come in complex conjugate pairs, it is required
that $f(x,y) = f(x,-y)$. Moreover the working of    \cite{BF23} for  GinSE it is required that $f(\mathbf r)$ be radially symmetric, in
which case the second term in (\ref{A.2}) is not present --- its stated form in presented as a conjecture.

However the dependence on $\beta$ no longer is analogous to that in Gaussian Hermitian ensembles if one asks
about local properties in the spectrum bulk (a local property is when the scaling is such that the mean spacing
between eigenvalues is unity, while the spectrum bulk is a portion of the spectrum away from the outer boundary of
the support of the spectrum, and away too from the real axis). Then explicit computation of the corresponding correlation
functions (see \cite{BF23}) are universal, with the same limiting determinantal point process appearing in each of
the three cases (for finite $N$ only the GinUE is determinantal, with both the GinOE and GinSE being Pfaffian point
processes).

In the pioneering work of Dyson \cite{Dy62c} attention was drawn to certain invariances associated with
the Gaussian ensembles. For example, the GOE is invariant with respect to conjugation with respect
to orthogonal matrices, which in fact is the reasoning for naming Gaussian orthogonal ensemble (GOE).
This in turn was shown to be associated with a time reversal symmetry (simply complex conjugation,
which is an involution for the GOE, or the operator $Q_{2N} K$, where $Q_{2N}$ is the $2N \times 2N$ anti-symmetric matrix
with all entries $-1$ ($+1$) in the leading upper (lower) triangular diagonal, and $K$ corresponds to
complex conjugation for  GSE). The theory of symmetric spaces identifies 10 distinguished matrix ensembles
from such a symmetry viewpoint \cite{KS99,AZ97}. Subsequent work on non-Hermitian (or non-unitary)
matrix ensembles from an analogous viewpoint \cite{BL02,Ma07,KSUS19} identifies no fewer that 38 distinct
cases. However, this number reduces down to 9 distinct cases if there is no more than a single symmetry involved
in the classification of \cite{KSUS19}; see \cite{HKKU20}. For these 9 distinct cases, large
scale numerical simulations of the bulk nearest neighbour spacing distributions were performed in \cite{HKKU20}, giving
compelling evidence that all but two of these 9 cases belonged to the GinUE universality class.
The exceptional ensembles from this viewpoint were found to be that of complex symmetric matrices,
and matrices of the form $M= Q_{2N} A$, where $A$ is a $2N \times 2N$ anti-symmetric matrix. This latter
class of non-Hermitian matrices have the property that $Q_{2N} M^T  Q_{2N}^T = M$, which is the condition
to be complex self (quaternion) dual, and moreover their spectrum is doubly degenerate; their introduction in
random matrix theory can be traced to Hastings \cite{Ha01} (see too the follow up work in
\cite[\S 4]{Fo16}). Sommers et al.~\cite{SFT99}, motivated
by problems in scattering theory, had a year or two earlier in 1999
initiated the study of   complex symmetric matrices
from a random matrix theory viewpoint. In the classification schemes referenced above, complex
symmetric matrices are labelled as AI${}^\dagger$, while complex self dual matrices are labelled
as AII${}^\dagger$.

While the eigenvalue statistics of GinUE can be computed analytically, in particular the bulk correlation
functions \cite{Gi65} and the nearest neighbour spacing distribution \cite{GHS88},
this is no longer the case
for non-Hermitian complex symmetric matrices, or non-Hermitian complex self dual matrices
(although one should be aware of the accurate determination of the spacing distribution via simulation in
\cite{HKKU20} and \cite{AMP22}).
 On the other hand, if
one asks instead about other statistical properties, in particular averaged products of characteristic polynomials,
two very recent works \cite{AAKP24,KKR24}, due to Akemann et al. and  Kulkarni, et al.~respectively,
show that progress  on the analytic front can be made. In fact these same results can be deduced by
appropriately specialising results contained in the further recent work of Liu and Zhang \cite{LZ24}.
Denote the set of $N \times N$ complex symmetric matrices by $\mathcal S_N(\mathbb C)$, and the set of
$2N \times 2N$ self dual matrices by $\mathcal Q \mathcal A_{2N}(\mathbb C)$ (here our choice of naming is following
the construction of the self dual matrices as the product of $Q_{2N}$ times a $2N \times 2N$ complex
anti-symmetric matrix).
 
The work of Akemann et al.~\cite{AAKP24} contains evaluation formulas for the average of two distinct characteristic
polynomials, where the complex conjugate is taken in one, for both $\mathcal S_N(\mathbb C)$ and 
$\mathcal Q \mathcal A_{2N}(\mathbb C)$ when chosen according to a Gaussian measure.

\begin{prop}\label{P1}
Denote by G$\mathcal S_N$ (G$\mathcal Q \mathcal A_{2N}$) the ensembles of random matrices from
the sets  $\mathcal S_N(\mathbb C)$ 
($\mathcal Q \mathcal A_{2N}(\mathbb C)$) chosen with measure proportional to $e^{-{\rm Tr} X X^\dagger}$
($e^{-{\rm Tr} X X^\dagger/2}$). Let
 \begin{equation}\label{A.3}
 E_N(x) = \sum_{j=0}^N {x^j \over j!}.
 \end{equation}  
We have
  \begin{equation}\label{A.3a}
 {1 \over C}  \langle \det (z \mathbb I_N - X) (\bar{w} \mathbb I_N - \bar{X}) \rangle_{{\rm G} \mathcal S_N} 
  = 
  \Big ( E_N(2  y) - {2 y \over N + 1}  E_{N-1}(2 y) \Big ) \Big |_{y = z \bar{w}}
 \end{equation}  
 and
 \begin{equation}\label{A.3b}
{1 \over C} \langle \det (z \mathbb I_{2N} - X) (\bar{w} \mathbb I_{2N} - \bar{X}) \rangle_{{\rm G}  \mathcal Q \mathcal A_{2N}}
=
 {N! \over (2N)!} \sum_{j=0}^N {(2j)! \over j!} (4 y)^{N-j} E_{2j}(2 y)) \Big |_{y = z \bar{w}},
 \end{equation}    
 where the normalisation $C$ is chosen so that the left hand sides are equal to unity for $  z = \bar{w} = 0$.
\end{prop}

The result of Proposition \ref{P1} generalises the evaluation formula for the same average with respect to GinUE \cite{APS09}
 \begin{equation}\label{A.3c}
 {1 \over C}   \langle \det (z \mathbb I_N - X) (\bar{w} \mathbb I_N - \bar{X}) \rangle_{{\rm GinUE}_N} 
 =  E_N(  z \bar{w} ).
  \end{equation}   
As noted in \cite{APS09}, an  alternative way to write (\ref{A.3c}) is in the form of a so-called duality relation (see the recent review \cite{Fo24+}) 
 \begin{equation}\label{A.3c+}  
  \langle \det (z \mathbb I_N - X) (\bar{w} \mathbb I_N - \bar{X}) \rangle_{{\rm GinUE}_N}  =  
 \langle (z \bar{w} + |u|^2)^N \rangle_{u \in \mathcal{CN}(0,1)},
  \end{equation}
  where $\mathcal{CN}(0,1)$ denotes the standard complex normal distribution.
  A  duality formula for the average of several products of GinUE random matrices had  been
  found earlier by Nishigaki and Kamenev \cite{NK02}, although all factors were absolute values squared.
  This constraint has been removed in the recent thesis of Serebryakov \cite[second statement of Th.~4.1]{Se23},
  to obtain the generalisation of (\ref{A.3c+})
\begin{equation}\label{6.V}
\Big  \langle \prod_{l=1}^k \det(z_l \mathbb I_N - Z)   \det({w}_l \mathbb I_N - \bar{Z}) \Big \rangle_{Z \in {\rm GinUE}_{N}} =
\bigg \langle \det \begin{bmatrix} D_1 & -Y \\ Y^\dagger & D_2 \end{bmatrix}^N \bigg \rangle_{Y \in {\rm GinUE}_{k}},
 \end{equation}
 where $D_1, D_2$ are diagonal matrices with diagonal entries $\mathbf z=(z_1,\dots,z_k)$ and $\mathbf w=(w_1,\dots,w_k)$ respectively.
 The very recent work of Kulkarni et al.~\cite{KKR24} generalises (\ref{6.V}) for averages over 
 G$\mathcal S_N$ and G$\mathcal Q \mathcal A_{2N}$.

\begin{prop}\label{P2}
Let $D_1,D_2$ be as above, We have
\begin{equation}\label{6.0v+}
\Big  \langle \prod_{l=1}^k \det(z_l \mathbb I_N - Z)   \det(\bar{w}_l \mathbb I_N - \bar{Z}) \Big \rangle_{Z \in {\rm G} \mathcal S_N}  =
\bigg \langle \det \begin{bmatrix} D_1 \otimes \mathbb I_2  & -Y \\ Y^\dagger & \bar{D}_2 \otimes \mathbb I_2 \end{bmatrix}^{N/2} \bigg \rangle_{Y \in {\rm GinSE}_k},
 \end{equation}
 Similarly
 \begin{equation}\label{6.0w+}
\Big  \langle \prod_{l=1}^k \det(z_l \mathbb I_{2N} - Z)^{1/2}   \det(\bar{w}_l \mathbb I_{2N} - \bar{Z})^{1/2} \Big \rangle_{Z \in {\rm G}  \mathcal Q \mathcal A_{2N}}  =
\bigg \langle \det \begin{bmatrix} D_1 & -Y \\ Y^\dagger & \bar{D}_2 \end{bmatrix}^{N} \bigg \rangle_{Y \in {\rm GinOE}_k}.
 \end{equation}
 
\end{prop}

\begin{remark}\label{R1} 
The matrix on the RHS of (\ref{6.0v+}) has doubly degenerate eigenvalues. The square root operation therein is defined by
including each eigenvalue only once in the implied product. Similarly for the matrix on the LHS of (\ref{6.0w+}). In fact in \cite{KKR24} 
this latter possibility is not considered, but rather the identity is stated as
 \begin{equation}\label{6.0W}
\Big  \langle \prod_{l=1}^k \det(z_l \mathbb I_{2N} - Z)   \det(\bar{w}_l \mathbb I_{2N} - \bar{Z}) \Big \rangle_{Z \in {\rm G}  \mathcal Q \mathcal A_{2N}}  =
\bigg \langle \det \begin{bmatrix} D_1   \otimes \mathbb I_2 & -Y \\ Y^\dagger & \bar{D}_2  \otimes \mathbb I_2  \end{bmatrix}^{N} \bigg \rangle_{Y \in {\rm GinOE}_{2k}}.
 \end{equation}
 This follows from (\ref{6.0w+}) with  $k \mapsto 2k$, $\mathbf z \mapsto (z_1,z_1,\dots,z_k,z_k)$ and $\mathbf w \mapsto (w_1,w_1,\dots,w_k,w_k)$.
 \end{remark} 

The derivations of the identities of Propositions \ref{P1} and \ref{P2} in \cite{AAKP24} and \cite{KKR24}
use supersymmetric Grassmann integrations methods. An alternative 
derivation of (\ref{6.V}), which can be applied too to the identities of  Proposition \ref{P2},
 uses matrix diffusion equations \cite{Gr16,LZ24}.
The derivation of \cite{Se23} uses matrix integrals from the theory of zonal polynomials \cite{Ta84,FR09}. Relative to the other two
methods, this has the advantage of not being restricted to the Gaussian class of orthogonally invariant measures
 \cite{Se23,SSD23,SS24,Fo24+}. Here we show that it
is similarly true that the dualities of Proposition  \ref{P2} can be derived in the context of integration formulas from
the theory of zonal polynomials, and that this allows for generalisations beyond the Gaussian case. 

We begin in Section \ref{S2.1} by introducing some essential features of zonal polynomial theory, in particular
some key integration formulas. These are used in Section \ref{S2.2} to derive the duality identities of Proposition
\ref{P2} using zonal polynomials. The context and some extensions of the duality identities are discussed in Section \ref{S3}.
First, we note a specialisation to powers of the characteristic polynomial. We compare both the dualities of
Proposition \ref{P2}, and those  obtained by  specialisating to powers of the characteristic polynomial, to known
dualities for GinOE and GinSE. This comparison is carried out at a quantitative level in Proposition \ref{P3.1},
revealing a decreasing (increasing)
profile as the magnitude of the argument of the characteristic polynomial increases towards the boundary of support.
The remainder
of Section \ref{S3} considers the case of complex symmetric and complex self dual matrices chosen with a so-called
 spherical measure, where in particular an analogous comparison result is derived.

\section{Zonal polynomial derivation of Proposition \ref{P2}}
\subsection{Some assumed theory}\label{S2.1}
Zonal polynomials are particular cases of Jack polynomials $P_\kappa^{(\alpha)}(\mathbf x)$
 \cite[Ch.~VI.6]{Ma95}, \cite[Ch.~12]{Fo10}, \cite[Ch.~7]{KK09}.
 Jack polynomials are symmetric polynomials of $N$ variables
 $\mathbf x = (x_1,\dots,x_N)$, labelled by a partition $\kappa :=
 (\kappa_1,\dots,\kappa_N)$ with parts which are non-negative
 integers ordered as $\kappa_1 \ge \kappa_2 \ge \cdots \ge \kappa_N$,
 and dependent on a parameter $\alpha > 0$. They can be specified as
 the polynomial eigenfunctions of the
differential operator
\begin{equation}\label{4.0}
\sum_{j=1}^N
x_j^2 {\partial^2 \over \partial x_j^2}  
+ {2 \over \alpha}  \sum_{1 \le j < k \le N}{1 \over x_j -
x_k}
\Big (x_j^2 {\partial \over \partial x_j} - x_k^2 {\partial
\over
\partial x_k}
\Big ),
\end{equation}
with the triangular structure with respect to the monomial basis  $\{ m_\mu(\mathbf x) \}$ 
\begin{equation}\label{4.0a}
P_\kappa^{(\alpha)}(\mathbf x) = m_\kappa(\mathbf x) +
\sum_{\mu < \kappa} c_{\kappa, \mu}^{(\alpha)} 
m_\mu (\mathbf x).
\end{equation}
In (\ref{4.0a})  the 
notation $\mu < \kappa$ denotes the partial order on partitions, defined by the requirement that 
 for each $s=1,\dots,\ell(\kappa)$ (where $\ell(\kappa)$  denotes the number of nonzero parts of $\kappa$)
 we have
$\sum_{i=1}^s \mu_i \le \sum_{i=1}^s \kappa_i$. The Jack polynomials furthermore have the homogeniety property
\begin{equation}\label{H}
c^{|\kappa|} P_\kappa^{(\alpha)}(\mathbf x) = P_\kappa^{(\alpha)}(c\mathbf x).
\end{equation}

The zonal polynomials are the special cases $\alpha = 1$ 
(this corresponding to the Schur polynomials $s_\kappa(\mathbf x)$),
$\alpha = 2$ (this, up to normalisation, corresponding the zonal polynomials
of mathematical statistics \cite{Mu82}), and also the case $\alpha = 1/2$.
 For general $\alpha > 0$ the Jack
polynomials satisfy the so-called
dual Cauchy identity (see e.g.~\cite[Eq.~(12.187)]{Fo10})
\begin{equation}\label{SM2}
\prod_{k=1}^N \prod_{l=1}^p (1 - x_k y_l) = \sum_{\kappa \subseteq (p)^N} (-1)^{|\kappa|} P_\kappa^{(\alpha)}(\mathbf x)  P_{\kappa'}^{(1/\alpha)}(\mathbf y), 
\end{equation}
where $(N)^p$ denotes the partition with $p$ parts all equal to $N$ and $|\kappa| = \sum_{l=1}^N \kappa_l$.
The notation $\kappa'$ denotes the conjugate partition, when the rows and columns of the corresponding
diagram are interchanged. 

Often used, for example in the mathematical statistics literature in the
case $\alpha = 2$, are the renormalised Jack polynomials
 \begin{equation}\label{15.ckp}
C_\kappa^{(\alpha)}(\mathbf x) := {\alpha^{|\kappa|} |\kappa|! \over
h_\kappa'} P_\kappa^{(\alpha)}(\mathbf x).
\end{equation}
Here $h_\kappa'$, and a companion quantity
$h_\kappa$, depends on $\alpha$ although this is repressed from the
notation. They are
defined in terms of the diagram of $\kappa$ according to
 \begin{equation}\label{2.2h}
h_\kappa' = \prod_{s \in \kappa} ( \alpha (a(s) + 1) + l(s)), \quad
h_\kappa = \prod_{s \in \kappa} (\alpha a(s) + l(s)+1).
\end{equation}
The quantities $a(s), l(s)$ are the arm and leg lengths at position $s$ in the diagram associated with $\kappa$ \cite{Ma95}.
One notes the inter-relation with respect to  the conjugate partition
 \begin{equation}\label{2.2h+}
 h_{\kappa'}'  = \alpha^{|\kappa|}  h_\kappa  |_{\alpha \mapsto 1/\alpha}.
\end{equation} 
We denote by ME${}_{\beta,N}[w]$ the probability density function (PDF) proportional to
$$
\prod_{l=1}^N w(x_l) \prod_{1 \le j < k \le N} |x_k - x_j|^\beta.
$$
In terms of $C_\kappa^{(\alpha)}$, for this class of PDF
with $w(\lambda) = \lambda^a e^{-\lambda} \mathbbm 1_{x > 0}$ ($a > -1$),
the general Jack polynomials also satisfy the integration formula
(see e.g.~\cite[Eq.~(12.152)]{Fo10})
 \begin{equation}\label{7.X1}
 \langle C_\kappa^{(\alpha)}(\mathbf x) \rangle_{{\rm ME}_{2/\alpha,N}[\lambda^a e^{-\lambda}]} =
C_\kappa^{(\alpha)}((1)^N) [a+1+ (N-1)/\alpha]_\kappa^{(\alpha)}.
 \end{equation} 
 The notation $(1)^N$ specifies  that $x_i = 1$ ($i=1,\dots,N$) in the argument, while
 \begin{equation}\label{5.1}
[u]_\kappa^{(\alpha)} = \prod_{j=1}^N {\Gamma(u - (j-1)/\alpha + \kappa_j) \over
\Gamma(u - (j-1)/\alpha) },
\end{equation}
is a generalisation of the classical Pochhammer symbol. Involving (\ref{5.1}) and the
quantity $ h_\kappa$ from (\ref{2.2h}) is the formula \cite[Prop.~12.6.2]{Fo10}
 \begin{equation}\label{5.1f}
 P_\kappa^{(\alpha)}((1)^N) = {\alpha^{|\kappa|} [N/\alpha]_\kappa^{(\alpha)} \over h_\kappa}.
\end{equation}

Particular to the zonal polynomial cases is an integration formula involving
non-Hermitian matrices \cite{Ta84},  \cite{FR09}.

\begin{prop} 
For $\alpha = 2 ,1$ and $1/2$, let $X,A,B$ be $N \times N$ matrices with real, complex and
quaternion entries respectively. In the case of quaternion entries require that the eigenvalues
of both $A$ and $B$ be doubly degenerate.
Let $X$ be a random matrix chosen from a left and right
unitary invariant distribution, with the unitary matrices drawn from the orthogonal group,
unitary group and symplectic unitary group respectively. We have
\begin{equation}\label{16.jlX}
\langle
C_\kappa^{(\alpha)}(A X^\dagger B X) \rangle_X =
{C_\kappa^{(\alpha)}(A)C_\kappa^{(\alpha)}(B) \over (C_\kappa^{(\alpha)}((1)^{N}))^2} 
\langle
C_\kappa^{(\alpha)}(X X^\dagger) \rangle_X,
\end{equation}
where $C_\kappa^{(\alpha)}(A X^\dagger B X)$ is
defined as $C_\kappa^{(\alpha)}(\mathbf y)$ with $\mathbf y = (y_1,\dots,y_{N})$ denoting the
eigenvalues of $A  X^\dagger B  X$, and similarly the meaning of
$C_\kappa^{(\alpha)}(A)$, $C_\kappa^{(\alpha)}(B)$. In the quaternion case all these matrices
are required to be doubly degenerate, and only one copy of the eigenvalues is to be included.
\end{prop}

We require too matrix integration formulas over the set of complex symmetric matrices, and
complex self dual matrices \cite{FS09}.

\begin{prop} 
Let $A,B$ be $N \times N$ complex symmetric ($2N \times 2N$ complex self dual) matrices. Let $X \in \mathcal S_N(\mathbb C)$ ($\mathcal Q \mathcal  A_{2N}(\mathbb C)$) 
be chosen from
a measure unchanged by $X \mapsto U X U^T$ for $U \in U(N)$ ($X \mapsto U X U^D$ for $U \in U(2N)$ where $U^D = Q_{2N} U^T Q_{2N}^T$.
 One has 
\begin{align}\label{t.1}
\Big \langle C_\kappa^{(2)}(A X)  C_\mu^{(2)}(B X^\dagger )  \Big \rangle_{X \in \mathcal S_N(\mathbb C)}
& = \delta_{\kappa,\mu} {C_{\kappa}^{(2)}(AB) \over s_{2\kappa}((1)^N)} \langle C_\kappa^{(2)}(X X^\dagger) \rangle_{X\in \mathcal S_N(\mathbb C)}, \nonumber \\
\Big \langle C_\kappa^{(1/2)}(A X)  C_\mu^{(1/2)}(B X^\dagger )  \Big \rangle_{X \in \mathcal Q \mathcal A_{2N}(\mathbb C)}
& =\delta_{\kappa,\mu}{C_{\kappa}^{(1/2)}(AB) \over s_{\kappa^2}((1)^{2N})} \langle C_\kappa^{(1/2)}(X X^\dagger) \rangle_{X\in \mathcal Q \mathcal A_{2N}(\mathbb C)}. 
\end{align}
Here the notation $2 \kappa$ denotes the partition with each part each $2 \kappa_i$, while the notation $\kappa^2$ denotes the partition in
which each part of $\kappa$ is repeated.
\end{prop}

\begin{remark}
The quantities $ s_{2\kappa}((1)^N),  s_{\kappa^2}((1)^{2N})$ appearing above can be substituted according to
\cite[Eqns.~(2.62), (2.63)]{Se23}
\begin{equation}\label{t.2}
 s_{\kappa^2}((1)^{2N}) = {1 \over | \kappa |!} C_\kappa^{(1/2)}(\mathbb I_N) [2N - 1]_\kappa^{(1/2)}, \quad
 s_{2 \kappa}((1)^N) = {1 \over | \kappa |!} C_\kappa^{(2)}(\mathbb I_N) [N/2 + 1/2]_\kappa^{(2)}.
 \end{equation}
 \end{remark} 
 
 \subsection{Proof of Proposition \ref{P2}}\label{S2.2}
 Consider first (\ref{6.0v+}).
Use  (twice) of the dual Cauchy identity (\ref{SM2})
with $\alpha = 2$
on the LHS  gives for this the expanded form
\begin{equation}\label{7.X}
\prod_{l=1}^k (w_l z_l)^N \sum_{\mu, \kappa \subseteq (k)^N} (-1)^{|\mu| + | \kappa|} P_{\mu'}^{(1/2)}(1/\mathbf z) P_{\kappa'}^{(1/2)}(1/  \bar{\mathbf w})
\langle P_\kappa^{(2)}(X) P_\mu^{(2)}(X^\dagger) \rangle_{X \in  {\rm G} \mathcal S_N},
 \end{equation}
 where the notation $1/\mathbf z$ (similarly $1/\mathbf w$) is used to denote the vector with entries the reciprocal of the entries of $\mathbf z$.
 The average herein can be evaluated by making use of the first matrix integral in (\ref{SM2}), with the RHS further simplified by use of 
(\ref{15.ckp}),  (\ref{7.X1}) and (\ref{t.2}) to read \cite[Prop.~2.15]{Se23} 
\begin{equation}\label{7.X1a} 
\langle P_\kappa^{(2)}(A X) P_\mu^{(2)}( B X^\dagger) \rangle_{X \in  {\rm G} \mathcal S_N} = \delta_{\kappa, \mu} P_\kappa^{(2)}(AB)  h_{\kappa'}|_{\alpha = 1/2}.
\end{equation}
Substituting (\ref{7.X1a}) in (\ref{7.X}), replacing $\kappa$ by $\kappa'$ in the summand, and making use of 
(\ref{5.1f}) with $\alpha = 2$, $N=k$,
shows that the LHS of (\ref{6.0v+}) is equal to
\begin{equation}\label{7.X2}
\prod_{l=1}^k (w_l z_l)^N \sum_{ \kappa \subseteq (N)^k}  P_{\kappa}^{(1/2)}(1/\mathbf z) P_{\kappa}^{(1/2)}(1/\bar{\mathbf w}) P_{\kappa'}^{(2)}((1)^N) 2^{-|\kappa|}
{[2k]_\kappa^{(1/2)} \over   P_{\kappa}^{(1/2)}((1)^k)}.
\end{equation}

A number of terms in (\ref{7.X2}) can be recognised as an example of the group integral (\ref{16.jlX})
with $\alpha = 1/2$ with $N \mapsto k$ and $X$ drawn from GinSE${}_k$. Thus, with $\tilde{D}_i :=
D_i  \otimes \mathbb I_2$, ($i=1,2$),
\begin{align}\label{7.X3}
\langle P_\kappa^{(1/2)}(\tilde{D}_1^{-1}  X^\dagger \bar{\tilde{D}}_2^{-1}   X) \rangle_{X \in {\rm GinSE}_k}& =
{ P_{\kappa}^{(1/2)}(1/\mathbf z)   P_{\kappa}^{(1/2)}(1/\bar{\mathbf w}) 
\over (P_{\kappa}^{(1/2)}(1)^k))^2  } 
  \langle  P_\kappa^{(1/2)}(X X^\dagger) \rangle_{X \in {\rm GinSE}_k} \nonumber \\
  & = { P_{\kappa}^{(1/2)}(1/\mathbf z)   P_{\kappa}^{(1/2)}(1/\bar{\mathbf w}) 
\over (P_{\kappa}^{(1/2)}(1)^k))  }  2^{-{|\kappa|}} [2k]_\kappa^{(1/2)}, 
\end{align}
where the second line follow by first changing variables $W = X X^\dagger$, 
then further changing variables to the eigenvalues and eigenvectors of $W$,
using the fact that the independent eigenvalues (they are otherwise doubly
degenerate) have PDF given by ME${}_{4,k}[\lambda e^{-2\lambda}]$
(see \cite[Prop.~3.2.2 with $m=n=k$, $\beta = 4$]{Fo10}), then making use of
(\ref{7.X1}). Consequently, we have that (\ref{7.X2}) is equal to
\begin{equation}\label{7.X4}
\prod_{l=1}^k (w_l z_l)^N \sum_{ \kappa \subseteq (N)^k}   P_{\kappa'}^{(2)}((1)^N)
\langle P_\kappa^{(1/2)}(\tilde{D}_1^{-1} X^\dagger \bar{\tilde{D}}_2^{-1} X) \rangle_{X \in {\rm GinSE}_k} 
 = \prod_{l=1}^k (w_l z_l)^N \Big \langle \prod_{l=1}^k ( 1 +  \gamma_l)^N \Big \rangle_{ {\rm GinSE}_k} ,
 \end{equation}
where the equality follows from   the dual Cauchy identity (\ref{SM2}), and $\{\gamma_l \}$ are
independent eigenvalues of the otherwise twice degenerate matrix $\tilde{D}_1^{-1} X^\dagger \bar{\tilde{D}}_2^{-1} X$.

Consider now the RHS of  (\ref{6.0v+}). With $A,B,C,D$ matrices of size $2k \times 2k$ blocks, use of the
 determinant identity
\begin{equation}\label{7.X5} 
\det   \begin{bmatrix} A & B \\ C & D \end{bmatrix}   = \det AD \det (\mathbb I_{2k} - A^{-1} B D^{-1} C),
   \end{equation}
   with the final determinant herein written as a product over its eigenvalues, allows for a rewrite.
    Taking into consideration
   that the eigenvalues are doubly degenerate then allows  the RHS of (\ref{7.X4}) to be identified.
   
   Following the general scheme of the above proof also provides a derivation of (\ref{6.0w+}).
For this we begin by making using of the dual Cauchy identity (\ref{SM2}) with
 $\alpha = 1/2$
on the LHS to give the expanded form
\begin{equation}\label{7.X+}
\prod_{l=1}^k (w_l z_l)^N \sum_{\mu, \kappa \subseteq (k)^N} (-1)^{|\mu| + | \kappa|} P_{\mu'}^{(2)}(1/\mathbf z) P_{\kappa'}^{(2)}(1/  \bar{\mathbf w})
\langle P_\kappa^{(1/2)}(X) P_\mu^{(1/2)}(X^\dagger) \rangle_{X \in  {\rm G}  \mathcal Q \mathcal S_{2N}}.
 \end{equation}   
 Next we make use of the second matrix integral in (\ref{SM2}), further simplified to read \cite[Prop.~2.17]{Se23} 
\begin{equation}\label{7.X1+} 
\langle P_\kappa^{(1/2)}(A X) P_\mu^{(1/2)}( B X^\dagger) \rangle_{X \in  {\rm G} \mathcal Q  \mathcal S_{2N}} = \delta_{\kappa, \mu} P_\kappa^{(1/2)}(AB)  h_{\kappa'}|_{\alpha = 2},
\end{equation}
substitute in (\ref{7.X+}) and replace $\kappa$ by $\kappa'$ in the summand. With this done, and noting $h_\kappa|_{\alpha = 2} = 2^{|\kappa|} [k/2]_\kappa^{(2)}/
P_\kappa^{(2)}((1)^k)$ as follows from (\ref{5.1f}) we have that the LHS of  (\ref{6.0w+}) is equal to
\begin{equation}\label{7.Y+}
\prod_{l=1}^k (w_l z_l)^N \sum_{\mu, \kappa \subseteq (N)^k} P_{\kappa}^{(2)}(1/\mathbf z) P_{\kappa}^{(2)}(1/  \bar{\mathbf w})
P_{\kappa'}^{(1/2)}((1)^N) {2^{|\kappa|} [k/2]_\kappa^{(2)} \over
P_\kappa^{(2)}((1)^k)};
\end{equation}
cf.~(\ref{7.X2}).

 The group integral (\ref{16.jlX})
with $\alpha = 2$ with $N \mapsto k$ and $X$ drawn from GinOE${}_k$ allows us to substitute
\begin{equation}\label{7.Y3}
P_{\kappa}^{(2)}(1/\mathbf z) P_{\kappa}^{(2)}(1/  \bar{\mathbf w})
 {2^{|\kappa|} [k/2]_\kappa^{(2)} \over
P_\kappa^{(2)}((1)^k)} = 
\langle P_\kappa^{(2)}({D}_1^{-1}  X^\dagger \bar{{D}}_2^{-1}   X) \rangle_{X \in {\rm GinOE}_k} .
\end{equation}
The dual Cauchy identity (\ref{SM2}) can now be used to reduce (\ref{7.Y+}) to the form
\begin{equation}\label{7.Y4}
 \prod_{l=1}^k (w_l z_l)^N \Big \langle \prod_{l=1}^k ( 1 +  \gamma_l)^N \Big \rangle_{ {\rm GinOE}_k} ,
 \end{equation}
 where here $\{\gamma_l \}$ are the eigenvalues of ${D}_1^{-1}  X^\dagger \bar{{D}}_2^{-1}   X$.
 Considering now the RHS of (\ref{6.0w+}) with the determinant therein rewritten according to 
 (\ref{7.X5}) gives precisely this expression, thus establishing the validity of the identity.
 
 \section{Discussion}\label{S3}
 Setting each $z_i$ and $w_i$ equal to $z$, and making use too of (\ref{7.X5}), we see that the RHS of
 each identity in Proposition \ref{P2} can be simplified.
 
 \begin{cor}\label{C1}
 We have
\begin{equation}\label{7.U1} 
\Big  \langle |  \det(z \mathbb I_N - Z)  |^{2k} \Big \rangle_{Z \in {\rm G} \mathcal S_N}  =
\langle \det ( |z|^2 \mathbb I_k + W )^N \rangle_{W \in {\rm ME}_{4,k}[\lambda e^{-2 \lambda}]},
\end{equation}
and
\begin{equation}\label{7.U2} 
\Big  \langle |  \det(z \mathbb I_{2N} - Z)  |^{k} \Big \rangle_{Z \in  {\rm G}  \mathcal Q \mathcal A_{2N}}  =
\langle \det ( |z|^2 \mathbb I_k + W )^N \rangle_{W \in {\rm ME}_{1,k}[\lambda^{-1/2} e^{- \lambda/2}]}.
\end{equation}
\end{cor} 

\begin{proof}
Proceeding as already described in the proof of Proposition \ref{P2} gives for the respective right hand sides
$$
\Big \langle \prod_{l=1}^k (|z|^2 \mathbb I_k + \gamma_l)^N \Big \rangle_{Y \in {\rm GinSE}_k}, \quad
\Big \langle \prod_{l=1}^k ( |z|^2 \mathbb I_k + \gamma_l)^N \Big \rangle_{Y \in {\rm GinOE}_k},
$$
where $\{\gamma_l\}$ are the independent eigenvalues of $Y Y^\dagger$ (in the first case the
eigenvalues of $Y Y^\dagger$ are doubly degenerate).
We now change variables $W = Y Y^\dagger$, which as already noted in the context of the proof of
Proposition  \ref{P2} gives rise to the ensembles  ${\rm ME}_{4,k}[\lambda e^{-2 \lambda}]$ and
$ {\rm ME}_{1,k}[\lambda^{-1/2} e^{- \lambda/2}]$ as in the stated identities.
\end{proof}

Dualities very similar to (\ref{7.U1}) and (\ref{7.U2}) are known from earlier literature. Thus \cite[set $\Sigma = \mathbb I_N$ in
Corollary 3]{FR09},
\begin{align}\label{7.U3}
\Big  \langle   \det(z \mathbb I_N - Z)^{2k} \Big \rangle_{Z \in {\rm GinOE}_N}  & =
\langle \det ( |z|^2 \mathbb I_k + W )^N \rangle_{W \in {\rm ME}_{4,k}[ e^{- \lambda}]}, \nonumber \\
\Big  \langle   \det(z \mathbb I_{2N} - Z)^{k} \Big \rangle_{Z \in {\rm GinSE}_N}  & =
\langle \det ( |z|^2 \mathbb I_k + W )^N \rangle_{W \in {\rm ME}_{1,k}[e^{- \lambda}]}.
\end{align}
These identities also follow from recently derived companion dualities to those of  Proposition \ref{P2} \cite[Eqns.~(4.62) and (4.96)]{Se23}
(see too \cite{LZ24})
\begin{equation}\label{5.47}
\Big \langle \prod_{j=1}^{2k} \det (G - z_j \mathbb I_N) \Big \rangle_{G \in {\rm GinOE}_N} =
\Big \langle \det \begin{bmatrix} X & Z \\ -Z & X^\dagger  \end{bmatrix}^{N/2} \Big \rangle_{X \in {\rm G} \mathcal A_{2k}} 
\end{equation}
and
\begin{equation}\label{5.48}
\Big \langle \prod_{j=1}^{2k} \det (G - z_j \mathbb I_{2N}) \Big \rangle_{G \in {\rm GinSE}_N} =
\Big \langle \det \begin{bmatrix} X & Z \\ -Z & X^\dagger  \end{bmatrix}^{N} \Big \rangle_{X \in {\rm G} \mathcal S_{2k}}. 
\end{equation}
This is seen by specialising the left hand sides to each $z_i$ equal,
and following the  analogous working to that for the proof of Corollary  \ref{C1} in relation to the right hand sides.

We know that global scaling for GinOE${}_N$ and GinSE${}_N$ corresponds to $Z \mapsto {1 \over \sqrt{N}} Z$,
with corresponding density then being given by the circular law ${1 \over \pi} \mathbbm 1_{|z|<1}$. Global scaling
for ${\rm G} \mathcal S_N$ ($ \mathcal Q \mathcal A_{2N}$) to give the same circular density is
$Z \mapsto \sqrt{2 \over N} Z$ ($Z \mapsto \sqrt{1 \over 2 N} Z$) \cite{AAKP24}. The dualities with this scaling
allow for the large $N$ limiting form of the ratio of characteristic polynomials, comparing (\ref{7.U1}) and
(\ref{7.U2}) with their companions in (\ref{7.U3}).

 \begin{prop}\label{P3.1}
 Suppose $|z| < 1$. Define
  \begin{equation}\label{Ws}
 W_{\beta,n}(a)  = \int_0^\infty dx_1 \cdots  \int_0^\infty dx_n \, \prod_{l=1}^n x_l^a e^{-x_l} \prod_{1 \le j < k \le n} |x_k - x_j|^\beta
  \end{equation}
  (this is a limiting case of the Selberg integral, which has an evaluation as a product of gamma functions; see \cite[Prop.~4.7.3]{Fo10}).
 We have
 \begin{equation}\label{K1}
 \lim_{N \to \infty} \Big ( \Big  \langle |  \det(z \mathbb I_N - Z)  |^{2k} \Big \rangle_{Z \in \sqrt{2 \over N} {\rm G} \mathcal S_N}  \Big /
 \Big  \langle   \det(z \mathbb I_N - Z)^{2k} \Big \rangle_{Z \in  {1 \over \sqrt{N}}  {\rm GinOE}_N} \Big ) = { W_{4,k}(0) \over W_{4,k}(1) }
 (1 - |z|^2)^k
 \end{equation}
 and
 \begin{equation}\label{K2}
 \lim_{N \to \infty} \Big ( \Big  \langle |  \det(z \mathbb I_N - Z)  |^{k} \Big \rangle_{Z \in \sqrt{1 \over 2 N}  {\rm G} \mathcal Q \mathcal A_{2N}}  \Big /
 \Big  \langle   \det(z \mathbb I_N - Z)^{k} \Big \rangle_{Z \in  {1 \over \sqrt{N}}  {\rm GinSE}_N}  \Big ) =  { W_{1,k}(0) \over W_{1,k}(-{1\over 2}) } (1 - |z|^2)^{-{k \over 2}}.
 \end{equation}
 \end{prop}
 
 \begin{proof}
 Consider (\ref{K1}). According to the duality  (\ref{7.U1}) and the first duality in (\ref{7.U3}), the task can be recast
 as that of calculating
   \begin{equation}\label{Ws1}
  \lim_{N \to \infty} \Big ( 
 \langle \det ( |z|^2 \mathbb I_k + W )^N \rangle_{W \in {\rm ME}_{4,k}[\lambda e^{- N \lambda}]} /
\langle \det ( |z|^2 \mathbb I_k + W )^N \rangle_{W \in {\rm ME}_{4,k}[ e^{- N \lambda}]} \Big ).
 \end{equation}
Define by $W_{\beta,n}(a)[f]$ the multiple integral (\ref{Ws}) with a further factor of $f=f(x_1,\dots,x_n)$ in the integrand.
Then we have that (\ref{Ws1}) is equal to
   \begin{equation}\label{Ws2}
 { W_{4,k}(0) \over W_{4,k}(1) }    \lim_{N \to \infty}  { W_{4,k}(1) \Big [ \prod_{l=1}^k e^{-(N-1)x_l} (|z|^2 + x_l)^N \Big ] \over    W_{4,k}(0) \Big [   \prod_{l=1}^k  e^{-(N-1)x_l}(|z|^2 + x_l)^N \Big ] }.
 \end{equation} 
 To compute the limit $N \to \infty$ we 
 apply Laplace's method to the integrands. A simple calculation shows that in each integration variable $x_l$ the maximum occurs at $1 - |z|^2$. Expanding about this
 point to second order in the exponent gives
 $$
 (1 - |z|^2)^k e^{-N \log (1 - |z|^2)}  \int_{-\infty}^\infty dx_1 \cdots  \int_{-\infty}^\infty dx_k \, \prod_{l=1}^k  e^{-N x_l^2/2} \prod_{1 \le j_1 < j_2 \le k} |x_{j_1} - x_{j_2}|^4
 $$ 
 for the numerator, and the same expression without the factor $(1 - |z|^2)^k$ for the denominator, and so (\ref{K1}) results.
 The derivation of (\ref{K2}) is similar.
\end{proof} 

\begin{remark}\label{R31} ${}$ \\
1.~The result of Proposition \ref{P3.1} is, for $k=1$ in (\ref{K1}), and for $k=2$ in (\ref{K2}), with results obtained in
\cite[Prop.~3.1]{AAKP24}. In the case of the denominators, these cases relate to the computation of
the eigenvalue density, which we know is the circular law for GinOE and GinSE. Thus relative to these base
cases, at a qualitative level, there is a decreasing (increasing) profile for $ {\rm G} \mathcal S_N$ and
$ {\rm G} \mathcal Q \mathcal A_{2N}$ respectively as $|z|$ approaches the
boundary of support \cite{AAKP24,KKR24}. In the Coulomb gas model relating to the Ginibre ensemble
\cite[\S 4]{BF24}, such behaviour is known to distinguish the regimes $\beta < 2$ from $\beta > 2$
\cite{CFTW14,CFTW15,CSA21}. \\
2.~The asymptotics of moments of the (absolute value of) the characteristic polynomial for GinUE matrices, for
which the duality \cite[with $\Sigma = \mathbb I_N$ in Corollary 3]{FR09}
\begin{equation}\label{GE}
\Big  \langle   |\det(z \mathbb I_N - Z)|^{2k} \Big \rangle_{Z \in {\rm GinUE}_N}   =
\langle \det ( |z|^2 \mathbb I_k + W )^N \rangle_{W \in {\rm ME}_{2,k}[ e^{- \lambda}]},
\end{equation}
cf.~(\ref{7.U3}), holds have attracted interest for their relevance to Gaussian multiplicative chaos
\cite{WW19}, their relationship to  Painlev\'e transcendents \cite{DS22}, from the viewpoint
of the Coulomb gas with a point charge insertion \cite{BSY24}. There are also works considering
the  asymptotics of moments of the characteristic polynomial for GinOE matrices (see e.g.~\cite{FT21,Af22}, with the most
recent being \cite{Ki24}. Using (\ref{GE}) and the first duality in (\ref{7.U3}), the method of proof of Proposition \ref{P3.1}
shows that for $|z|<1$
\begin{equation}\label{GE1}
\lim_{N \to \infty} \Big ( \Big  \langle |  \det(z \mathbb I_N - Z)  |^{2k} \Big \rangle_{Z \in  {1 \over \sqrt{N}}  {\rm GinOE}_N}  \Big /
 \Big  \langle   \det(z \mathbb I_N - Z)^{2k} \Big \rangle_{Z \in  {1 \over \sqrt{N}}  {\rm GinUE}_N} \Big ) = 1.
\end{equation}
This was remarked on in \cite{Ki24} as following from results in \cite{Af22} (and therefore holding beyond the
assumption of Gaussian entries). We note that his does not contradict the fluctuation formula (\ref{A.2}) as
$ |  \det(z \mathbb I_N - Z)  |^{2k} = \exp \Big (2k \sum_{l=1}^N \log | z - z_l | \Big )$, where $\{z_l\}$
are the eigenvalues of $Z$. The linear statistic so exhibited thus has a singularity at $z$, which for
$z$ inside the  eigenvalue support  nullifies the
applicability of (\ref{A.2}).
\end{remark}

As noted in the Introduction (see below Remark \ref{R1}), an advantage of zonal polynomial theory relative to
supersymmetric Grassmann integration methods is that they apply to a wider class of measures on matrix
space with a unitary invariance, beyond a Gaussian. On this point, there is a particular non-Gaussian measure on the sets
$\mathcal S_N(\mathbb C)$ and $\mathcal A_N(\mathbb C)$ (this denoting $N \times N$ complex anti-symmetric
matrices) which are already known in the literature from the viewpoint of duality identities \cite{Se23}. This measure
is defined to have a PDF on the matrix spaces proportional to
\begin{equation}\label{7.U5} 
\det (\mathbb I_N + X X^\dagger)^{-N-K}.
\end{equation}
We will denote complex symmetric matrices chosen with this PDF as ${\rm S} \mathcal S_{N,K}$ 
(here the ``S'' standards for spherical; see e.g.~\cite{FF11} for justification).
With $N$ replaced by $2N$, we will consider this PDF in relation to complex self dual matrices from
$\mathcal Q \mathcal A_{2N}(\mathbb C)$ too, which we denote as ${\rm S} \mathcal Q \mathcal A_{2N,K}$.
One recalls \cite[Exercises 3.6 q.3]{Fo10} that for complex matrices (\ref{7.U5}) results by forming
$(A^\dagger A)^{-1/2} B$, where $A$ is a $K \times N$ ($K \ge N$), and $B$ is an $N \times N$ standard
complex Gaussian matrix. However, such a construction no longer holds for matrices from
${\rm S}  \mathcal S_{N,K}$ and ${\rm S} \mathcal Q \mathcal A_{2N,K}$. 

Here we take up the task
of deriving the analogues of the identities of Corollary \ref{C1} in this case. Assumed Jack polynomial
theory for this task is the average \cite{Wa05},  \cite{FS09}
\begin{equation}\label{W}
\langle  C_\kappa^{(\alpha)}(\mathbf x)  \rangle_{{\rm ME}_{2/\alpha,N}[\lambda^{b_1}(1+\lambda)^{-b_1 - b_2 - 2 - 2(N-1)/\alpha} ]}
  =  C_\kappa^{(\alpha)}((-1)^N) 
{[b_1 + 1 + (N-1)/ \alpha]_\kappa^{(\alpha)} \over [-b_2]_\kappa^{(\alpha)} },
\end{equation}
and the series expression for a Jack polynomial based generalisation of the Gauss hypergeometric function
(see \cite[Ch.~13]{Fo10})
\begin{equation}\label{3.40}
  {\vphantom{F}}_2^{\mathstrut} F_1^{(\alpha)}(a_1,a_2;b_1; \mathbf x):=\sum_\kappa  {1 \over |\kappa|!} 
 \frac{[a_1]^{(\alpha)}_\kappa [a_2]^{(\alpha)}_\kappa }{[b_1]^{(\alpha)}_\kappa}
C_\kappa^{(\alpha)}(\mathbf x). 
\end{equation}
In a special case, when all the arguments of $\mathbf x$ are equal, we have the matrix average form \cite[Eq.~(13.7)]{Fo10}
\begin{equation}\label{7.Zc}  
 {\vphantom{F}}_2^{\mathstrut} F_1^{(\alpha)}(-r,-b;-c;(s)^N) =  \Big \langle \prod_{l=1}^r  (1 - x_l s)^N \Big \rangle_{\mathbf x \in {\rm ME}_{2 \alpha,r}[\lambda^{a_1}(1 - \lambda)^{a_2} \mathbbm 1_{0<x<1}]}
 \bigg |_{\substack{a_1 =  - 1 +(b - r+1)\alpha \\ a_2 =  -1 + (c - b - r+1)\alpha}},
 \end{equation}
 valid provided $a_1,a_2 > -1$ and $r$ is a positive integer.

 \begin{prop}\label{C1}
 For $K$ large enough so that the right hand sides are well defined we have
\begin{multline}\label{7.V1} 
\Big  \langle |  \det(z \mathbb I_N - Z)  |^{2k} \Big \rangle_{Z \in {\rm S} \mathcal S_{N,K}}  = |z|^{2N} {\vphantom{F}}_2^{\mathstrut} F_1^{(2)}(-k,-k;- K+1; (-1/|z|^2)^N) \\
=
\langle \det ( |z|^2 \mathbb I_k + W )^N \rangle_{W \in {\rm ME}_{4,k}[\lambda (1 - \lambda)^{2(K-2k)-1} \mathbbm 1_{0<x<1}]},
\end{multline}
and
\begin{multline}\label{7.V2} 
\Big  \langle |  \det(z \mathbb I_{2N} - Z)  |^{k} \Big \rangle_{Z \in  {\rm S}  \mathcal  Q \mathcal A_{2N,K}}  =
|z|^{2N} {\vphantom{F}}_2^{\mathstrut} F_1^{(1/2)}(-k,-k;- 2K-2; (-1/|z|^2)^N) \\
= \langle \det ( |z|^2 \mathbb I_k + W )^N \rangle_{W \in {\rm ME}_{1,k}[\lambda^{-1/2} (1 - \lambda)^{K-k+1/2} \mathbbm 1_{0<x<1}]}.
\end{multline}
\end{prop} 

\begin{proof}
A result from Jack polynomial theory gives (see e.g.~\cite[Eq.~(13.4)]{Fo10})
    \begin{equation}\label{7.E6}  
    \det (\mathbb I_N - Q)^{k} =   \sum_{ \kappa \subseteq (N)^{r}} { [-k]_\kappa^{(\alpha)} \over | \kappa|!} C_\kappa^{(\alpha)}(Q).
  \end{equation}  
  In the case of (\ref{7.V1}), we use this with $\alpha = 2$ to obtain for the LHS
   \begin{multline}\label{7.E7}  
|z|^{2N} \sum_{\mu, \kappa \subseteq (N)^k}  |z|^{-2 |\kappa|}
  \Big (  { [-k]_\kappa^{(2)} \over | \kappa|!} \Big )^2
\langle C_\kappa^{(2)}(Z) C_\mu^{(2)}(Z^\dagger) \rangle_{Z \in  {\rm S} \mathcal S_{N,K}} \\
= |z|^{2N} \sum_{\kappa \subseteq (N)^k}   |z|^{-2 |\kappa|} { ([-k]_\kappa^{(2)})^2 \over | \kappa|!} {1 \over [(N+1)/2]_\kappa^{(2)}}
\langle C_\kappa^{(2)}(Z Z^\dagger) \rangle_{Z \in  {\rm S} \mathcal S_{N,K}}, 
 \end{multline}
 where the equality follows from use of the first integration formula in (\ref{t.1}) and the second identity in
 (\ref{t.2}).  
 
 For the average on the RHS of (\ref{7.E7}), we recall the PDF (\ref{7.U5}), and change variables to the
 eigenvalues of $W = Z Z^\dagger$, making use of the knowledge that the corresponding Jacobian is
proportional to  $\prod_{1 \le j < k \le N} | x_k - x_j|$ \cite[Table 3.1]{Fo10}. This shows
  \begin{equation}\label{7.E8}  
\langle C_\kappa^{(2)}(Z Z^\dagger) \rangle_{Z \in  {\rm S} \mathcal S_{N,K}} =
\langle    C_\kappa^{(2)}(\mathbf x) \rangle_{\mathbf x \in {\rm ME}_{1,N}[1/(1 + x)^{N+K} \mathbbm 1_{x > 0}]} =
 C_\kappa^{(2)}((-1)^N)  {[ (N+1)/ 2]_\kappa^{(2)} \over [-K + 1]_\kappa^{(2)} },
\end{equation}
where the second equality follows from (\ref{W}). Substituting in (\ref{7.E7}) and using the homogeniety property (\ref{H})
we recognise from (\ref{3.40}) the first equality of (\ref{7.V1}). The second equality follows from  (\ref{7.Zc}).

We turn our attention now to (\ref{7.V2}). The first point to note is that 
  \begin{equation}\label{7.E9}  
|  \det(z \mathbb I_{2N} - Z)  |^{k} = 
 \prod_{l=1}^N (z - \gamma_j)^k (\bar{z} - \bar{\gamma}_j)^k ,
 \end{equation}
 where $\{ \gamma_j \}$ are the independent eigenvalues of the doubly degenerate matrix $Z$. From this starting point,
 we make use of (\ref{7.E6}) (twice) now with $\alpha = 1/2$, then the second integration formula in (\ref{t.1}) and the first identity in
 (\ref{t.2}) to obtain for the LHS 
  \begin{equation}\label{7.E9a}    
  |z|^{2N} \sum_{\kappa \subseteq (N)^k}   |z|^{-2 |\kappa|} { ([-k]_\kappa^{(2)})^2 \over | \kappa|!} {1 \over [2N-1]_\kappa^{(2)}}
\langle C_\kappa^{(1/2)}(Z Z^\dagger) \rangle_{Z \in  {\rm S}   \mathcal  Q \mathcal A_{2N,K}}. 
 \end{equation} 
 To compute the average in this expression we require the fact that the Jacobian for the change of variables
 to the independent eigenvalues of $Z Z^\dagger$ for $ \mathcal  Q \mathcal A_{2N}$ is
 proportional to  $\prod_{1 \le j < k \le N} | x_k - x_j|^4$ \cite[Table 3.1]{Fo10}. Using too (\ref{W}) it follows
 \begin{equation}\label{7.E10}  
\langle C_\kappa^{(1/2)}(Z Z^\dagger) \rangle_{Z \in  {\rm S}  \mathcal  Q \mathcal A_{2N,K}} =
\langle    C_\kappa^{(1/2)}(\mathbf x) \rangle_{\mathbf x \in {\rm ME}_{4,N}[1/(1 + x)^{2 N+K} \mathbbm 1_{x > 0}]} =
 C_\kappa^{(1/2)}((-1)^N)  {[ 2N-1]_\kappa^{(1/2)} \over [-2K - 2]_\kappa^{(1/2)} }.
\end{equation} 
Substituting in (\ref{7.E9a}), the first equality in (\ref{7.V2}) now follows by using the homogeniety property (\ref{H})
and comparing with (\ref{3.40}), and the second equality follows from (\ref{7.Zc}).

  \end{proof}
  
  \begin{remark}
  In the study of quantum conductance through a normal metal-superconductor junction, one encounters
  the matrix Bogoliubov-deGennes equation (see e.g.~\cite{Be97}) reads
   \begin{equation}\label{S1a}
  \begin{bmatrix} A & B \\ - \bar{B} & - A^T \end{bmatrix} \begin{bmatrix} \mathbf u \\ \mathbf v \end{bmatrix} = \lambda
   \begin{bmatrix} \mathbf u \\ \mathbf v \end{bmatrix},
   \end{equation} 
   where $A,B$ are square matrices with $B = - B^T$, $A = A^\dagger$. The classification 
   from the viewpoint of time reversal symmetry and spin-rotation invariance
   of the Hamiltonian
   herein undertaken in \cite{AZ97} (see also \cite[\S 3.2.2]{Fo10}) leads to four basic subcases. Two of these
   are shown to be equivalent under similarity transformation to the Hamiltonian
     \begin{equation}\label{S1b}
   \begin{bmatrix} 0_{N \times N} & C \\ C^\dagger & 0_{N \times N} \end{bmatrix}
   \end{equation} 
   where $C$ is either complex symmetric or complex self dual (in which case $N$ must be even). These
   are referred to as types CI and DIII respectively. The tie in with the
   above proof is the squared eigenvalues of (\ref{S1b}) (which come in $\pm$ pairs) are equal to the
   eigenvalues of $C C^\dagger$, with the noted  eigenvalue Jacobians then prominent in the development of the
   theory.
   \end{remark}

  From the text about (\ref{7.U5}), in terms of the distinct eigenvalues of $X X^\dagger$, $\{\gamma_j\}_{j=1}^N$ say,
  the defining PDF for matrices from ${\rm S} \mathcal S_{N,K}$ and ${\rm S} \mathcal Q \mathcal A_{2N,K}$
  is proportional to $\prod_{l=1}^N(1+\gamma_l)^{-N-K}$ and $\prod_{l=1}^N(1+\gamma_l)^{-4 N-2 K}$ respectively.
By way of comparison, consider  random matrices formed according to $X = (A^\dagger A)^{-1/2} B$,
where $A$ ($B$) are $(N+K) \times N$, ($N \times N$) matrices with real (quaternion) standard Gaussian elements. 
Let us denote these ensembles by ${\rm SrOE}_{N,K}$ and ${\rm SrSE}_{N,K}$ respectively (here the
``Sr'' indicates spherical).
Then we have that the PDF specifying the matrices is proportional to $\prod_{l=1}^N(1+\gamma_l)^{-N-K/2}$ \cite{GN99}
and $\prod_{l=1}^N(1+\gamma_l)^{-4N-2K}$ \cite{MP17}. Hence we see that ${\rm S} \mathcal S_{N,K}$ and
${\rm SrOE}_{N,2K}$ are directly comparable, as are ${\rm S} \mathcal Q \mathcal A_{2N,K}$ and
${\rm SrSE}_{N,K}$. Moreover, results from \cite{Fo24+} give (assuming $K$ is large enough so that the right hand
sides are well defined)
\begin{multline}\label{7.V1x} 
\Big  \langle |  \det(z \mathbb I_N - Z)  |^{2k} \Big \rangle_{Z \in {\rm SrOE}_{N,2K}}  = |z|^{2N} {\vphantom{F}}_2^{\mathstrut} F_1^{(2)}(-k,-k+1/2;- K+1/2; (-1/|z|^2)^N) \\
=
\langle \det ( |z|^2 \mathbb I_k + W )^N \rangle_{W \in {\rm ME}_{4,k}[ (1 - \lambda)^{2(K-2k)+1} \mathbbm 1_{0<x<1}]},
\end{multline}
and
\begin{multline}\label{7.V2x} 
\Big  \langle |  \det(z \mathbb I_{2N} - Z)  |^{k} \Big \rangle_{Z \in  {\rm SrSE}_{N,K}}  =
|z|^{2N} {\vphantom{F}}_2^{\mathstrut} F_1^{(1/2)}(-k,-k-1;- 2K-1; (-1/|z|^2)^N) \\
= \langle \det ( |z|^2 \mathbb I_k + W )^N \rangle_{W \in {\rm ME}_{1,k}[\lambda^{-1/2} (1 - \lambda)^{K-k-1/2} \mathbbm 1_{0<x<1}]}.
\end{multline}
This allows for ratios of the comparable averages to be computed for $N \to \infty$, analogous to 
Proposition \ref{P3.1}.

\begin{prop}
 We have
 \begin{equation}\label{K1x}
 \lim_{N \to \infty} \Big ( \Big  \langle |  \det(z \mathbb I_N - Z)  |^{2k} \Big \rangle_{Z \in  {\rm S} \mathcal S_{N,K}}  \Big /
 \Big  \langle   \det(z \mathbb I_N - Z)^{2k} \Big \rangle_{Z \in   {\rm SrOE}_{N,2K}} \Big ) \propto 
 (1 + |z|^2)^{-2k}
 \end{equation}
 and
 \begin{equation}\label{K2x}
 \lim_{N \to \infty} \Big ( \Big  \langle |  \det(z \mathbb I_N - Z)  |^{k} \Big \rangle_{Z \in   {\rm S} \mathcal Q \mathcal A_{2N,K} } \Big /
 \Big  \langle   \det(z \mathbb I_N - Z)^{k} \Big \rangle_{Z \in   {\rm SrSE}_{N,K}}  \Big ) \propto (1 + |z|^2)^k .
 \end{equation}
 \end{prop}

\begin{proof}
The essential point is that in the averages on the RHS of (\ref{7.V1}),  (\ref{7.V2}),  (\ref{7.V1x}) and  (\ref{7.V2x}) the maximum
of the integrand for large $N$ occurs at the endpoints of the range of integration $x_l = 1$, so we expand
$$
\prod_{l=1}^N(|z|^2 + x_l) \approx \exp \Big ( N \log (|z|^2 + 1) +  \sum_{l=1}^N (x_l-1) {1 \over 1 + |z|^2} \Big ).
$$
A simple change of variables now gives the stated results.
\end{proof}

As commented on in Remark \ref{R31} in the case of a Gaussian measure, these ratios exhibit the qualitative feature of a decreasing (increasing)
profile as $|z|$ increases.

 \subsection*{Acknowledgements}
 Helpful feedback on an earlier draft by Gernot Akemann and Mario Kieburg is appreciated.

\small
\providecommand{\bysame}{\leavevmode\hbox to3em{\hrulefill}\thinspace}
\providecommand{\MR}{\relax\ifhmode\unskip\space\fi MR }
\providecommand{\MRhref}[2]{%
  \href{http://www.ams.org/mathscinet-getitem?mr=#1}{#2}
}
\providecommand{\href}[2]{#2}

  \end{document}